\documentclass{llncs}
\usepackage{graphicx}
\usepackage{makeidx}

\usepackage{color,amsfonts,graphics,latexsym,amsmath,amssymb,graphicx,%subfigure,
alltt}
{\begin{list}{}%
         {\setlength{\leftmargin}{#1}}%
         \item[]%
}
{\end{list}}

\usepackage{epsfig}
\usepackage{tabularx}
\usepackage{latexsym}

\def\qed{$\Box$}
\usepackage{enumerate}

\begin{document}

\title{New Online Algorithm for Dynamic Speed Scaling with Sleep State}
\author{Gunjan Kumar, Saswata Shannigrahi}
\institute{Indian Institute of Technology Guwahati, India. \email{\{k.gunjan,saswata.sh\}@iitg.ernet.in} }
\maketitle

%.................Start of Abstract.........................

\begin{abstract} 
In this paper, we consider an energy-efficient scheduling problem where $n$ jobs $J_1, J_2, \ldots, J_n$
need to be executed such that the total energy usage of these jobs is minimized while ensuring that
each job is finished within it's deadline. We work
in an online setting where a job is known only at it's arrival time, along with
it's processing volume and deadline. In such a setting, the currently best-known algorithm by Han
et al. \cite{han} provides a competitive ratio max $\{4, 2 + {\alpha}^{\alpha}\}$ 
of energy usage. In this paper, we
present a new online algorithm SqOA which provides a competitive ratio  
max $\{4, 2 + (2-1/{\alpha})^\alpha 2^{\alpha-1}\}$ of energy usage. 
For $\alpha \geq 3$, the competitive ratio of our algorithm is better than that of
any other existing algorithms for this problem. 

\noindent \textbf{Keywords:} Energy efficient algorithm, online algorithm, scheduling algorithm
\end{abstract}

\section{Introduction}

Over the last few decades, the energy cost for computing has increased exponentially. 
For example, the electricity cost imposes a significant strain 
on the budget of data centers where CPUs account for $50$ to $60$ percent of the energy
consumption. The related problems include the heat generated, leading to a
reduced reliability of hardware components. As a result, it has become an important 
challenge for the algorithm researchers to design algorithms which are efficient with respect to
the energy usage by the CPU.

Yao, Demers and Shenker \cite{yao} were the first one to investigate the following scheduling problem
where the objective is to minimize the total energy consumed by the processing of the jobs. The general 
setting of their problem, on which we also work on, is explained below. 

Assume that a processor can run at any arbitrary speed, and consider $n$ jobs $J_1, \ldots, J_n$
that have to be processed by such a processor. Each job $J_i$ has a release time $r_i$, a 
deadline $d_i$ and a processing volume $w_i$. A job $J_i$ must be executed in the time interval 
$[r_i , d_i]$. The processing volume $w_i$ is the amount of work that must be completed to finish $J_i$. 
It's also assumed that the preemption of jobs is allowed, i.e. the processing of a job may be stopped and resumed later. If the processor is running at a speed $s$, the required power $P(s)$ is assumed to be equal to 
$s^{\alpha}$ where $\alpha > 1$ is a constant. (In practice, $\alpha = 3$.) The energy consumed by the execution of the algorithm is the amount of power integrated over the execution period of the algorithm. 

Under the above assumptions, Yao et al. \cite{yao} developed an $O(n^3)$-time algorithm (later improved to 
$O(n^2 \log n)$-time) to schedule $n$ jobs such that the total energy consumption is minimized. However, this algorithm (called YDS algorithm) works {\it offline}, i.e., the release time, deadline and the 
processing volume of each job is known
from the beginning. In an online version of the problem, a job $J_i$ is only known 
at the time of it's arrival, when the deadline $d_i$ and the processing volume $w_i$ are also
revealed. An online algorithm is called $c$-competitive if, for any job sequence, the total energy consumption of the algorithm is at most $c$ times that of an optimal offline algorithm for the same set of jobs.

In the same paper \cite{yao} as above, Yao et al. devised two online algorithms for the energy-efficient scheduling problem described above. The first algorithm, called {\it Average Rate}. has a competitive ratio at most $2^{\alpha - 1} {\alpha}^{\alpha}$, for any $\alpha \geq 2$. The competitive ratio of the second algorithm, called {\it Optimal Available (OA)}, is exactly ${\alpha}^{\alpha}$, and therefore the second algorithm is better
than the first one in terms of competitiveness. Bansal et al. \cite{bansal_kimbrel} provided an algorithm BKP which improved the competitive ratio to $2 (\frac{\alpha}{\alpha - 1})^{\alpha} e^{\alpha}$. For small values of $\alpha$, Bansal et al. \cite{bansal_chan_pruhs} provides a better competitive algorithm {\it qOA} whose competitiveness is at most $\frac{4^{\alpha}}{2 e^{\frac{1}{2}}{\alpha}^{\frac{1}{4}}}$.

Let us describe the ideas behind the two algorithms {\it Optimal Available (OA)} and {\it qOA} as these two
algorithms are relevant to the work done in this paper. The algorithm OA can be described as 
follows: whenever a new job arrives, the algorithm computes an optimal schedule for the currently available unfinished jobs. The optimal schedule for the currently available unfinished jobs can be computed using the YDS algorithm mentioned earlier. The algorithm {\it qOA} slightly modifies the algorithm {\it OA}.
It sets the speed of the processor to be $q \geq 1$ times the speed that the algorithm {\it OA} would run
in the current state. Setting the value of $q$ to be $2 - \frac{1}{\alpha}$, this algorithm achieves a competitive ratio of $\frac{4^{\alpha}}{2 e^{\frac{1}{2}}{\alpha}^{\frac{1}{4}}}$.

\subsection{Our contribution}

Irani et al. \cite{irani} started the investigation of the scenario
where a processor can be transitioned into a sleep state in which the energy consumption is $0$.
In the active state, the power consumption is given by the equation $P(s) = s^{\alpha} + g$,
where $g > 0$. One can note here that the power consumption is strictly greater than $0$ when the processor
is in an active state. However, it's not always beneficial to send a processor to the sleep state 
when it has no work to do, because $L > 0$ amount of wake-up energy is needed to 
wake-up the processor from the sleep state to the active state. Accordingly, an online algorithm needs
to decide when to send a processor to the sleep state, and when to wake-up the processor back to 
the active state, in addition to determining the speed at every moment during the active state.

Irani et al. \cite{irani} devised an online algorithm which is 
$(2^{2 \alpha - 2} {\alpha}^{\alpha} + 2^{\alpha -1} + 2)$ competitive for energy. Han et al. \cite{han} developed an algorithm SOA which improves the competitive ratio to $\{4, 2 + {\alpha}^{\alpha}\}$. 
In this paper, 
we combine the ideas used in the algorithms qOA \cite{bansal_chan_pruhs} and SOA \cite{han}, and provide an algorithm SqOA with a competitive ratio of max $\{4, 2 + (2-1/{\alpha})^\alpha 2^{\alpha-1}\}$. Note that our algorithm has a better competitive ratio than SOA for $\alpha \geq 3$. To summarize, we would be proving the following
theorem in this paper.

\begin{theorem}
\label{maintheorem}
SqOA achieves the competitive ratio max $\{4, 2 + (2-1/{\alpha})^\alpha 2^{\alpha-1}\}$ for any 
$\alpha > 1$.
\end{theorem}

\section{Algorithm SqOA}
\label{mainsection}

Let us start with some definitions and notations that we would be using in this paper. A processor 
is stated to be in a {\it working state} if it is currently running with a strictly positive speed. The processor is in an {\it idle state} if the processor is active but running with a speed $0$. Finally, the 
processor is said to be in a {\it sleep state} if it's no longer active.

Throughout this paper, we denote the notations used in Bansal et al. \cite{bansal_chan_pruhs}. 
The current time is always denoted as $t_{0}$. For $t_{0} \le t' \le t'' $, 
$w_{a}(t',t'')$ denotes the total amount of unfinished work for SqOA at $t_{0}$ that has a deadline during 
$(t', t'']$. The quantity $w_{0}(t',t'')$ is defined similarly for the optimal algorithm OPT. 
The current speeds of SqOA and OPT are denoted by $s_{a}$ and $s_{o}$, respectively.
 
Let $d(t' ,t'')$ = max $\{0, w_{a}(t',t'') - w_{0}(t',t'')\}$ be the excess unfinished work that 
SqOA has relative to OPT among the already released jobs with deadlines in the range $(t',t'']$, and let 
$\rho$ denote the current highest density. A sequence of critical
times $t_{0} < t_{1} < t_{2} < .... < t_{h}$ is defined iteratively as follows: 
let $t_{1}$ be the latest time such 
that $ d(t_{0}, t_{1})/( t_{1} − t_{0})$
is maximized. Note that $t_{1}$ is no more than the latest deadline of any job released thus far. If $t_{i}$ is earlier than the latest deadline, then let $t_{i+1} > t_{i}$ be the latest time, not later than the latest deadline, that maximizes $d(t_{i}, t_{i+1})/(t_{i+1} - t_{i})$. We refer to the intervals 
$[t_{i}, t_{i+1}]$ as critical intervals. We define $g_{i}$ as 
$g_{i} = d(t_{i},t_{i+1}) / (t_{i+1} - t_{i}) $. Note that $g_{0} , g_{1} ,....,g_{h} $ is a
non-negative strictly decreasing sequence. 
%Let $E_{w}$ be the working energy: Energy incurred in processing the jobs i.e. when the processor's speed is strictly greater than 0. %

Let $\rho$ be the current highest density, i.e., $\max_{t > t_{0}}  \frac {w_{a}(t_{0}, t )} {(t-t_{0})}$. 
The critical speed $s*$ is defined as the speed $s$ which minimizes $\frac{P(s)}{s}$. The algorithm SqOA
initially starts from a sleep state, and decides on the following actions based on the current state the processor is in.\\ \\
{\bf Case I: The processor is currently in the working state:} If $s* > \rho > 0 $, keep working on the job with the earliest deadline at speed $s*$. If $ \rho \ge s* $,  keep working on the job with the earliest deadline at speed q$\rho$.   \\ \\
{\bf Case II: The processor is currently in the idle state:} Let $ t' \le t $ be the last time in the working state. (Set $t' = 0$, if undefined). If $\rho \ge s* $ , switch to the working state; otherwise, if $(t-t') g = L$, then switch to the sleep state. \\  \\
{\bf Case III: The processor is currently in the sleep state:} If $\rho \ge s* $, switch to the working state.  \\ \\

Han et al. \cite{han} proved that if we bound the working energy of their algorithm SOA or any algorithm having same sleep and idle strategy as SOA in terms of working energy of OPT, then that would also bound the total energy (including the idle and wake-up energy) of online algorithm in terms of total energy of OPT. Their proof is based on the property that if SOA wakes up after sleep or idle state then it must be busy till the deadline of the job which the algorithm has started processing after idle/sleep state. Our algorithm SqOA has the same idle and sleep strategy as SOA and so we need to show that the property above also holds for SqOA. 

\begin{lemma}
\label{lemma1}
Algorithm SqOA is busy till the deadline of the job J which gets processed after idle/sleep interval of SqOA. 
\end{lemma}

\begin{proof}
Note that the processor switches to the working state only when $\rho \geq s*$. 
If $\rho = s*$ at the time $t_w$ of the switching, then the processor is set a
speed $q \rho$ at $t_w$. After an infinitesimal amount of time, the processor is set a speed $s*$ which is only infinitesimally greater than $\rho$ at that point of time. Therefore, the job $J$ can't finish before it's deadline.  \newline
If $\rho > s*$ at the time $t_w$ of the switching, then the processor is set a speed $q \rho$ at $t_w$.
After a finite amount of time, $\rho$ becomes infinitesimally smaller than $s*$ and the processor is set a speed $s*$. As before, the job $J$ can't finish before it's deadline in this case as well.
 \qed
\end{proof}
Based on this observation and the work of Han et al. \cite{han}, we can bound the total energy of SqOA by bounding the working energy. 

\begin{lemma}
\label{lemma2}    
If the working energy of  SqOA is at most $c$ times that of OPT, then its total energy is at most 
$\max\{c+2, 4\} $ times that of OPT. 
\end{lemma}
We adopt the method of amortized local competitiveness to find the competitive ratio of the algorithm. We choose our potential function as 
 \[  \phi = \beta \sum_{i = 0}^ {h-1} \Hat{g_{i}}^{\alpha-1}  (w_{a}(t_{i},t_{i+1}) - w_{o}(t_{i},t_{i+1}))\] 
where $ \hat{g_{i}} = \max \{s*,g_{i}\}$. This function $\phi(t)$  is a function of the current time and it must satisfy the following properties.
\begin{itemize}
\item $\phi(t)$ is $0$ at the arrival of first job and the deadline of last job.
\item Its value does not changes abruptly by discrete events like arrival of job, job completion by Optimal or Online algorithm and instantaneous changes in critical intervals.
\item At any time $t$, the following holds: $E_{a}(t) + \phi(t) \le c E_{o}(t)$  where $c$ is the competitiveness ratio.  
\end{itemize}

The following two lemmas are slight modifications of the results proved for the qOA algorithm given by 
Bansal et al. \cite{bansal_chan_pruhs}.

\begin{lemma}
\label{lemma3}
Job arrival do not change $ \phi $ and also do not change the critical intervals. Also completion of work by OPT or SqOA do not change $\phi $.
\end{lemma}
\begin{proof}
Upon a job arrival, the work of both the online and offline algorithms increase exactly by the same amount,
and hence the excess work $d(t', t'')$ does not change for any $t'$ and $t''$. Also, $\phi$ is a continuous function of the unfinished works of SqOA and OPT and the unfinished works continuously decrease to 0. 
\qed
\end{proof}
Note that the critical times may change due to SqOA or OPT working on the jobs. So, we must ensure that the instantaneous changes in a critical interval cannot cause an abrupt change in $\phi$. 
\begin{lemma}
\label{lemma4}
Instantaneous changes in critical intervals do not abruptly change $\phi $.
\end{lemma}
\begin{proof}
There are three ways the critical times can change.

\begin{itemize}
\item {\it Merging of two critical intervals:} As SqOA follows an EDF (earliest deadline first) policy,
it must work on jobs with deadlines in $[t_{0}, t_{1}]$, causing $g_{0}$ to decrease until it becomes equal to $g_{1}$. At this point, the critical intervals $[t_{0}, t_{1}]$ and $[t_{1}, t_{2}]$ merge together. The
potential function $\phi$ does not change by this merger as $g_{0} = g_{1} $ and thus 
$\hat g_{0} = \hat g_{1}$ at this point.  \\

\item {\it Splitting of a critical interval:} As OPT works on some job with deadline 
$t_{0} \in (t_{k}, t_{k+1}] $, the quantity 
\begin{displaymath}
\frac {w_{a}(t_{k} , t') - w_{o} (t_{k} , t')} {t'-t_{k}} 
\end{displaymath}  
may increase faster than
\begin{displaymath}
\frac {w_{a}(t_{k} , t_{k+1}) - w_{o} (t_{k} , t_{k+1}) } {t_{k+1}-t_{k}},
\end{displaymath}
causing this interval to split into two critical intervals, $[t_{k}, t']$ and $ [t', t_{k+1}]$. This split does not change $\phi$ as the density of the excess work for both of these newly formed intervals is $g_{k}$ and hence $\hat g_{k}$ is also the same. \\

\item {\it Formation of a new critical time:} This happens when a job arrives with a deadline later
than any of the previous jobs. The creation of the new critical time $t_{h+1}$ doesn't change the
potential $\phi$ because the excess unfinished work  $ (w_{a}(t_{h},t_{h+1}) - w_{o}(t_{h},t_{h+1}))$ 
is equal to $0$. \qed
\end{itemize}

\end{proof}
The observations above imply that the potential function does not change due to any discrete events
such as arrivals, job completions, or instantaneous changes in critical intervals. Then, in order to establish that SqOA is $c$-competitive with respect to energy, it is sufficient to show the following  condition at all times: 
\begin{displaymath}
\frac {dE_{a}(t)} {dt}  + \frac{d\phi}{dt} \le c \frac {dE_{o}(t)} {dt}.
\end{displaymath}  
To show this, we need a few more lemmas. Let $s_{o}' = \max_{t >t_{0}} \frac {w_{o}(t_{0},t)} {t-t_{0}}$. Then, the following relationship holds between $s_{o}'$, $s_o$ and $s*$.
\begin{lemma}
\label{lemma5}
If $ s_{o} > 0 $  then $ s_{o} \ge  s_{o}' $. Also if $s_{o}' \ge s*$ then $s_{o} \neq 0 $.
\end{lemma}
\begin{proof}
If $ s_{o} < s_{o}'$, then one of the jobs J contributing to $w_{o}(t_{o},t')$ must be scheduled at a speed greater than $s_{o}' $. Also  $s_{o} > 0 $ implies that some job (let us call it $J'$) is scheduled at
the current time. As the power function  $P(s) = s^\alpha + g $ is strictly convex, we can reduce the energy consumption of OPT by increasing the speed of $J'$, decreasing the speed of $J$ and scheduling it an interval corresponding to $J$. This contradicts the optimality of OPT. 

For the second part of the Lemma, let $s_{o}' \ge s* $. If $s_{o} = 0$, one of the jobs J (among those contributing to $w_{o}(t_{o},t')$) must be scheduled at a speed greater than $s_{o}'$ and 
hence with a speed $s*$. Note that $s*$ is the speed at which processing energy of a job is minimized and therefore the following relationship holds: $P(s1)/s1 > P(s2)/s2 > P(s*)/s* $ for $s1>s2>s*$. 
(The inequalities are strict because here P is strictly convex function). By decreasing the speed of J and scheduling it at current time $t_{o}$, we can further decrease energy consumption. This contradicts the optimality of OPT. \qed
\end{proof}
\begin{lemma}
\label{lemma6}
If $\rho \ge s*$, then $s_{a} \ge qg_{0}$.
\end{lemma}
\begin{proof}
By the definition of SqOA, $ \rho \ge s*$ implies that
\begin{displaymath}
s_{a} = q.max_{t>t_{0}} \frac {w_{a}(t_{0} , t)} {t-t_{0}} \ge q. \frac {w_{a}(t_{0},t_{1})} {t_{1} - t_{0}} \ge q. \frac { d(t_{0} ,t_{1})} {t_{1}-t_{0}} = qg_{0}. 
\end{displaymath} 
\end{proof}
\begin{lemma}
\label{lemma7}
Suppose $\rho \ge s* $. If  $s_{o} = 0$, then $ s_{a} \le qg_{0} +q s* $. If $s_{o} > 0$, 
then $s_{a} \le qg_{0} + qs_{o}$. 
\end{lemma}
\begin{proof}
By the definition of SqOA, $ \rho \ge s* $ implies that
\begin{equation*} s_{a} = q.max_{t>t_{0}} \frac {w_{a}(t_{0} , t)} {t-t_{0}}  
 \le q. max_{t > t_{0}} \frac {w_{o}(t_{0} , t) + d(t_{0},t)} {t-t_{0}}  
\end{equation*}
\begin{equation*}
 \le q. max_{t>t_{0}} \frac {w_{o}(t_{o} , t)} {t -t_{0}} + q .max_{t > t_{0}} \frac{d(t_{0} , t)} {t-t_{0}}  \\ 
\le q.s_{o}' + q .g_{0}.  \end{equation*}
From Lemma \ref{lemma5}, $s_{o} = 0$ implies that $s_{o}' < s*$. Therefore, $s_{a} \le qg_{o} +q s*$ in this case.  On the other hand, $ s_{o} > 0$ implies that $s_{o}' \le s_{o}$. Therefore, $s_{a} \le qg_{o} + qs_{o}$
in this case. 
\qed
\end{proof} 
We would now move toward proving that the equation $ \frac {dE_{a}(t)} {dt}  + 
\frac{d\phi}{dt} \le c \frac {dE_{o}(t)} {dt}$ is true at all times $t$ except at the time of arrival for a job. However, Lemma \ref{lemma3} guarantees that $E_{a}(t) + \phi(t) \le c E_{o}(t) $ is true at all time (including the arrival time of any job). Since $\phi $ vanishes at the end, the total working energy of SqOA would
be at most $c$ times that of OPT.
Let us start considering the different cases depending on the values of $\rho,g_{i},s_{o}$ and $s_{a}$ at the current time. For each of these cases, we would get an inequality in terms of q, $\beta$ and c. However, we don't know how to choose the values of q, $\beta$ and c which would satisfy all the equations and minimize c. Instead, we will take q = 2 - 1/$ {\alpha} $ , $\beta$ = c = $q^\alpha  2^{\alpha - 1} $ and show that all the inequalities are satisfied. Thus, the value of $c$ in the competitive ratio of our algorithm would be 
$c = q^\alpha  2^{\alpha - 1}$.  

Without any loss of generality, we can assume that both OPT and SqOA schedule jobs according to the
Earliest Deadline First policy, and therefore SqOA is working on a job with a deadline at 
most $t_{1}$. Let $t'$ be the
deadline of the job that OPT is working on, and let $k$ be such that $t_{k} < t' \le t_{k+1}$. \\
% We will now consider different cases depending on speed of OPT and  SqOA and critical intervals at current time and then show that eq. ( ) is true for all cases. \\

\noindent \textbf {Case 1: $g_{0} > s*  , k = 0 , s_{a} > 0 , s_{o} > 0$} \\ \\
%Without any loss of generality, we can assume that $g_{0} \ge s* $ in time interval $ [t_{0}, t_{0} + dt_{0}] $ because if $ g_{0} %$ = s* at $ t_{0} $ and  $g_{0} $ decreases in $ [t_{0}, t_{0} + dt_{0}] $, then we will consider this case corresponding to $g_{0} %\le s* $. \\ \\
Since $g_{0} > s*$, there exists an interval $[t_{0},t_{0}+dt_{0}]$ such that
$\hat{g_{0}} = \max \{s*,g_{0}\} = g_{0}$ at any point of time in $[t_{0},t_{0}+dt_{0}]$. Therefore, \\  \\
$ \frac {d\phi} {dt_{0}}  = \beta \frac {d} {dt_{0}} (g_{0}^{\alpha-1} (w_{a}(t_{0},t_{1}) -w_{o}(t_{0},t_{1}))) $ \\
$ = \beta [g_{0}^{\alpha-1}(-s_{a}+s_{o}) + (\alpha-1)g_{0}^{\alpha-2}d(t_{0},t_{1}) \frac {d} {dt_{0}} (\frac {w_{a}(t_{0},t_{1}) - w_{o}(t_{0},t_{1})} {t_{1} -t_{0}} )] $ \\
$ = \beta [g_{0}^{\alpha-1}(-s_{a}+s_{o}) + (\alpha-1)g_{0}^{\alpha-2}d(t_{0},t_{1}) (\frac {(t_{1}-t_{0}) (-s_{a}+s_{0}) + d(t_{0},t_{1})} {(t_{1}-t_{0})^2})] $ \\
$ = \beta (\alpha g_{0} ^{\alpha - 1} (-s_{a} + s_{o}) + (\alpha - 1) g_{0}^\alpha )$. \\ \\
As $ \frac {dE_{a}(t)} {dt} = s_{a}^\alpha + g$  and $\frac {dE_{o}(t)} {dt} = s_{o}^\alpha + g$, we need to show that  
 \[ s_{a}^\alpha + g + \beta (\alpha g_{0} ^{\alpha - 1} (-s_{a} + s_{o}) + (\alpha - 1) g_{0}^\alpha )  \le c ( s_{o}^\alpha + g)  \qquad \qquad (1) \]
% As $ g_{0} \le \rho $ so  $\rho < s* $ would imply $g_{0} < s* $. Hence $\rho $ must be at least s*.\\ \\

\noindent Equation (1) can be written as \\
$s_{a}^\alpha - \beta \alpha g_{0} ^{\alpha - 1} s_{a} + \beta \alpha g_{0} ^{\alpha - 1} s_{o} + \beta (\alpha - 1) g_{0}^\alpha - c s_{o}^\alpha \le g (c-1)$, 
which is implied if \\ \\
$ s_{a}^\alpha - \beta \alpha g_{0} ^{\alpha - 1} s_{a} + \beta \alpha g_{0} ^{\alpha - 1} s_{o} + \beta (\alpha - 1) g_{0}^\alpha - c s_{o}^\alpha \le 0. \qquad \qquad \qquad \qquad (1a)$ \\ \\
\noindent It can be easily seen that $\rho \ge g_{0} > s*$. It follows from Lemma \ref{lemma7} that 
$qg_{0} \le s_{a} \le qg_{0} + qs_{o} $.  Note that the left hand side of equation (1a) is a convex function in $s_{a}$, and therefore it only remains to prove the inequality for $ s_{a} = qg_{0} $ and $ s_{a} = q (g_{0} + s_{o}).$ \\ \\
Substituting $ s_{a} = qg_{0} $ in the left hand side of (1a), we get
\begin{equation*}
q^\alpha g_{0}^\alpha - \beta q \alpha g_{0}^\alpha + \beta \alpha g_{0}^{\alpha - 1} s_{o} + \beta (\alpha - 1) g_{0}^\alpha - c s_{0}^\alpha = 
(q^\alpha - \beta \alpha q + \beta (\alpha - 1))g_{0}^\alpha + \beta \alpha g_{0}^{\alpha - 1} s_{o} - c s_{o}^\alpha.
\end{equation*}
Taking derivative with respect to $s_{o}$, we get that this is maximized at $s_{o}$ satisfying $ c s_{o}^{\alpha-1} = \beta g_{0}^{\alpha -1}$ and hence $ s_{o} = g_{0}$. Substituting this for $s_{o}$ and
cancelling $ g_{0}^\alpha $, it follows that we need to satisfy the following equation: 
$(q^\alpha - \beta \alpha q + \beta (\alpha - 1)) + \beta (\alpha - 1) \le 0$. This is easily satisfied as
$q^\alpha  < \beta $. \\ \\
Substituting $s_{a} = q g_{0} + q s_{o} $, the inequality 1(a) becomes \\ \\
%\begin{displaymath}
$ q^\alpha (g_{0} + s_{o})^\alpha - \beta \alpha q (g_{0} + s_{o}) g_{0}^{\alpha -1} + \beta \alpha g_{0}^{\alpha -1} s_{o} +\beta (\alpha-1) g_{0}^\alpha - c s_{o}^\alpha  \le 0                             \qquad \qquad \qquad\\
%\end{displaymath} \\
%\begin{displaymath}
\Leftrightarrow q^\alpha (g_{0}+s_{o})^\alpha - \beta(q\alpha - (\alpha-1))g_{0}^\alpha - \beta \alpha(q-1) g_{0}^{\alpha-1} s_{o} -cs_{o}^\alpha \le 0. \qquad \qquad \qquad \qquad (1b)$ \\
%\end{displaymath} \\
Let $s_{o} = x · g_{0}$. It follows from inequality 1(b) that we need to satisfy 
%\begin{displaymath}
$ q^\alpha (1+x)^\alpha - \beta(q\alpha - (\alpha-1)) - \beta \alpha (q-1) x - cx^\alpha \le 0$. 
Since $\beta = c = q^{\alpha} 2^{\alpha - 1}$, this inequality would be satisfied if
$(1+x)^\alpha - \alpha 2^{\alpha-1} -2^{\alpha-1} x^\alpha \le 0$. 
The maximum value of the left hand side is attained at $x=1$ where the value of its derivative $ \alpha (1+x)^{\alpha-1} - 2^{\alpha-1} \alpha x^{\alpha - 1}$ equals 0. It can be easily seen that this maximum value is 
$(2^\alpha - \alpha 2^{\alpha-1} - 2^{\alpha-1})$ = $ 2^{\alpha-1} (1- \alpha ) < 0$.
 \\ 

\noindent \textbf {Case 2:  $g_{0} > s* , k>0,  g_{k} > s* , s_{a} > 0 ,s_{o} > 0$} \\ \\
\noindent Since $g_{0}, g_k > s*$, there exists an interval $[t_{0},t_{0}+dt_{0}]$ such that
$\hat{g_{0}} = \max \{s*,g_{0}\} = g_{0}$ and $\hat{g_{k}} = \max \{s*,g_{k}\} = g_{k}$ 
at any point of time in $[t_{0},t_{0}+dt_{0}]$. Therefore, \\  \\
$ \frac {d\phi} {dt_{0}}  = \beta [\frac {d} {dt_{0}} (g_{0}^{\alpha-1} (w_{a}(t_{0},t_{1}) -w_{o}(t_{0},t_{1}))) + \frac {d} {dt_{0}} (g_{k}^{\alpha-1} (w_{a}(t_{k},t_{k+1}) -w_{o}(t_{k},t_{k+1})))] $ \\ 
$=\beta [\alpha g_{0} ^{\alpha - 1} (-s_{a}) + (\alpha - 1) g_{0}^\alpha + \alpha g_{k}^{\alpha-1}(s_{o})]$. \\

\noindent Note that $g_{k} < g_{0}$, as we have observed that $g_0, g_1, \ldots$ is a strictly decreasing sequence. Therefore, the quantity $\frac {d\phi} {dt_{0}}$ above is less than or equal to 
$\beta (\alpha g_{0} ^{\alpha - 1} (-s_{a} + s_{o}) + (\alpha - 1) g_{0}^\alpha )$. Moreover, it can be easily seen that $\frac {dE_{a}(t)} {dt} = s_{a}^\alpha + g$  and $\frac {dE_{o}(t)} {dt} = s_{o}^\alpha + g$. Therefore, we need to prove exactly the same inequality (1) in case 1. \\

\noindent \textbf {Case 3:  $g_{0} > s* , k>0 $ , $g_{k} < s* , s_{a} > 0 ,s_{o} > 0 $} \\ \\
\noindent Since $g_{0}, g_k > s*$, there exists an interval $[t_{0},t_{0}+dt_{0}]$ such that
$\hat{g_{0}} = \max \{s*,g_{0}\} = g_{0}$ and $\hat{g_{k}} = \max \{s*,g_{k}\} = s*$ 
at any point of time in $[t_{0},t_{0}+dt_{0}]$. Therefore, \\  \\
$ \frac {d\phi} {dt_{0}}  = \beta [\frac {d} {dt_{0}} (g_{0}^{\alpha-1} (w_{a}(t_{0},t_{1}) -w_{o}(t_{0},t_{1}))) + \frac {d} {dt_{0}} ((s*)^{\alpha-1} (w_{a}(t_{k},t_{k+1}) -w_{o}(t_{k},t_{k+1})))]$
$=\beta [\alpha g_{0} ^{\alpha - 1} (-s_{a}) + (\alpha - 1) g_{0}^\alpha + (s*)^{\alpha-1}(s_{o})]$. \\

\noindent We need to show that
$  s_{a}^\alpha + g  +\beta [\alpha g_{0} ^{\alpha - 1} (-s_{a})$ + $(\alpha - 1) g_{0}^\alpha + (s*)^{\alpha-1}(s_{o})] \le c s_{o}^\alpha + cg$. As we have already proven the inequality (1), it suffices to prove 
that $\beta (s*)^{\alpha-1}s_{o} \le \beta \alpha g_{0}^{\alpha-1} s_{o} \Leftrightarrow g_{0} \ge \frac {s*} {\alpha^{\frac {1}{\alpha-1}}}$. This is true since $g_{0} > s* $. \\

\noindent \textbf {Case 4:  $g_{0} > s* , k>0 $ , $g_{k} = s* , s_{a} > 0 ,s_{o} > 0 $} \\ \\
\noindent Since $g_k = s*$, there exists an interval $[t_{0},t_{0}+dt_{0}]$ such that
either $\hat{g_{k}} = \max \{s*,g_{k}\} = g_{k}$, or $\hat{g_{k}} = \max \{s*,g_{k}\} = s*$
at any point of time in $[t_{0},t_{0}+dt_{0}]$. In the former case, the proof corresponds 
to case 2 because $\frac {d\phi} {dt}, \frac {dE_{a}(t)} {dt}, \frac {dE_{o}(t)} {dt}$
are exactly the same. In the later case, the proof corresponds to Case 3 for the same reason. \\

\noindent \textbf {Case 5: $g_{0} < s* , s_{a} > 0 ,s_{o} > 0 $} \\ \\
As we have observed that $g_0, g_1, \ldots$ is a strictly decreasing sequence, $g_{0} < s* \Rightarrow g_{k} < s* $. Therefore, $\frac{d\phi}{dt_{0}}$ remains the same for $k = 0$ and $k > 0$. Hence, we need
not consider them as different cases.\\ \\
As $ \frac {d\phi} {dt_{0}}  = \beta (s*)^{\alpha-1} (-s_{a}+s_{o}) $, we need to prove that \\ \\ 
$s_{a}^\alpha + g + \beta (s*)^{\alpha-1} (-s_{a}+s_{o}) \le c ( s_{o}^\alpha + g).$ \quad \quad \quad 
\quad \quad \quad \quad \quad \quad \quad \quad \quad \quad \quad (2)\\ \\
\noindent Consider the case  $\rho < s* $. Then, the algorithm SqOA runs at a speed $s_{a} = s*$.  So the inequality (2) becomes 
 $  (s*)^\alpha + g + \beta (s*)^{\alpha-1} (- s* + s_{o}) \le c ( s_{o}^\alpha + g). $ \\ \\
Let $s_{o} = xs* $ . Substituting this value of $s_{o}$, we obtain \\ \\
$(s*)^{\alpha} + g -\beta(s*)^{\alpha} + \beta(s*)^{\alpha} x \le c x^{\alpha} (s*)^{\alpha} + cg $ \\ 
%We know  $ (s*)^{\alpha} = \frac{g}{\alpha-1} $.
%Multiplying both side by $\frac{\alpha-1}{g}$ and since $\beta = c $ \\
$ \Leftrightarrow1+ (\alpha-1) -\beta  +\beta x \le cx^\alpha +  c (\alpha-1) $ \\
$ \Leftrightarrow \alpha +\beta x \le cx^{\alpha} + c\alpha $ \\
$ \Leftrightarrow cx^{\alpha}+(c-1)\alpha-cx \ge 0 $ \\ \\
Differentiating the left hand side of the above inequality w.r.t. $x$ and equating it to $0$, 
we obtain $x =  (\frac {1}{\alpha})^{\frac{1}{\alpha-1}}.$ The value of the left hand side at this 
value of $x$ is equal to $-c \frac{1}{\alpha^{\frac{1}{\alpha-1}}} + c \alpha -\alpha $ 
$ >-c . 1 +c \alpha -\alpha  $ 
$ = (c-1)(\alpha-1) - 1 > 0 $. Since the second derivative of $cx^{\alpha}+(c-1)\alpha-cx$ is
greater than $0$ at any value of $x$ (this is true since $\alpha > 1$), it implies that
$cx^{\alpha}+(c-1)\alpha-cx$ attains it's minimum at $x =  (\frac {1}{\alpha})^{\frac{1}{\alpha-1}}.$ This completes the proof of inequality (2) in this case.\\ \\
Next, consider the case $ \rho \ge s* $. From Lemma \ref{lemma7}, it implies that
$s_{a} \le qg_{0}+qs_{o}$. Since $s_a > 0$, the algorithm SqOA is in a working state and it
implies that $s_a \geq s*$. Note that we have already shown that the inequality (2) is true for $s_{a} = s*$. As the function $s_{a}^\alpha + g + \beta (s*)^{\alpha-1} (-s_{a}+s_{o})$ is convex, it only remains to show that the inequality (2) holds for $s_{a} = qg_{0}+qs_{o}$ as well. \\ \\ 
Consider $s_{a} = qg_{0}+qs_{o}$. We need to prove that \\ \\
$ q^\alpha(g_{0}+s_{o})^\alpha + g + \beta(s*)^{\alpha-1} (-q(g_{0}+s_{o}) +s_{0}) \le c s_{o}^\alpha +cg $. 
\\ \\
Let $s_{o} = x g_{0}$ and $g_{0} = y s*$. (Note that $y \le 1$.) Substituting this value of $s_{o}$
and $g_0$, we obtain the following inequality which we need to prove. \\ \\
$ q^\alpha y^\alpha (1+x)^\alpha  - \beta q (1+x) y +\beta xy \le cx^\alpha y^\alpha + (c-1)(\alpha-1) $  \qquad \qquad \qquad (3)\\ \\
Let $ h(x) =q^\alpha y^\alpha(1+x)^\alpha - cx^\alpha y^\alpha - \beta qy $. Note that 
the inequality (3) can now be rewritten as $h(x)- \beta qxy + \beta xy \leq (c-1)(\alpha -1)$.
Since $(- \beta qxy + \beta xy)$ is negative (it follows from the fact that $q>1$), it suffices to
prove that $h(x) \leq (c-1)(\alpha -1)$. \\ \\

\noindent Differentiating this function $h(x)$, we obtain \\ \\
$h'(x) =q^\alpha y^\alpha \alpha (1+x)^{\alpha-1}-cy^\alpha \alpha x^{\alpha-1} \geq 0$ \\
$ \Leftrightarrow (1+x)^{\alpha-1} \geq 2^{\alpha-1} x^{\alpha-1} \\ 
\Leftrightarrow$ $ x \leq 1.$ \\ \\
Therefore, $h(x)$ attains it's maximum at $x=1$. It can be easily calculated
that
$ h(1) = q^\alpha y^\alpha . 2^\alpha   -c y^\alpha  - \beta q y $ 
$ = q^\alpha y^\alpha (2^\alpha - 2^{\alpha-1}) - \beta qy $ 
$ = q^\alpha y^\alpha 2^{\alpha-1} - \beta qy $. \\ 
Therefore, the equation (3) will be true if  
$  q^\alpha y^\alpha 2^{\alpha-1} - \beta qy - (c-1) (\alpha-1) \le 0 $.  \\ \\
Note that the equation $ q^\alpha y^\alpha 2^{\alpha-1} - \beta qy - (c-1) (\alpha-1)$ is a convex function in 
$y$, where $ 0 \le y \le 1 $. Therefore, it suffices to show this inequality is true for y = 0 and y = 1 .\\ \\
If $y = 0$, we need to show that $ - (c-1) (\alpha-1) \le 0 $, and if $y = 1$, we need to show that
$ (q^\alpha 2^{\alpha-1} - q q^\alpha 2^{\alpha-1}) - (c-1) (\alpha -1 ) \le 0 $. Both these inequalities are
trivially true since $\alpha, c, q > 1$.\\

\noindent \textbf {Case 6: $s_{a} > 0 , s_{o} = 0 , g_{0} > s* $} \\ \\
Note that $\frac {d\phi} {dt_{0}}  = \beta (\alpha g_{0}^{\alpha-1}(-s_{a}) + (\alpha-1)g_{0}^\alpha)$. 
Therefore, we need to show that \\ \\
$ s_{a}^\alpha + g + \beta (\alpha g_{0}^{\alpha-1}(-s_{a}) + (\alpha-1)g_{0}^\alpha) \le 0. \qquad
\qquad \qquad \qquad \qquad (4)$\\ \\
Since $g_{0} > s* $, it implies that $\rho > s* $ and hence
%$ g_{0} \le \rho \le s* $ (not possible) \\ \\
  $ qg_{0} \le s_{a} \le qg_{0} + qs* $ by Lemma \ref{lemma7}.\\ \\
Consider $ s_{a} = q g_{0} $. Then, the inequality (4) becomes: \\ \\ 
$ q^\alpha g_{0}^\alpha + g - \beta \alpha q g_{0}^{\alpha} + \beta (\alpha-1) g_{0}^\alpha \le 0 $ \\
$ \Leftrightarrow (q^\alpha -\beta \alpha q + \beta (\alpha -1)) g_{0}^\alpha + g \le 0 $ \\
 $ \Leftrightarrow q^\alpha (1 - \alpha 2^{\alpha-1})g_{0}^\alpha + g \le 0 $. \\ \\
The inequality above would be true if $q^\alpha (1 - \alpha 2^{\alpha-1}) (s*)^\alpha + g \le 0 \Leftrightarrow q^\alpha (1 - \alpha 2^{\alpha -1 }) + \alpha - 1 \le 0$ which can be easily seen to be true. \\ \\
Next, consider $s_{a} = qg_{0} + qs*$ and $g_{0} = x s*$. (Note that $x \geq 1$.) Then, the inequality (4) becomes 
\\ \\ 
$ q^\alpha (1+x)^\alpha - \beta \alpha q (1+x) x^{\alpha-1} + \beta (\alpha -1 ) x^\alpha + \alpha -1 \le 0. $ \\
$ \Leftrightarrow q^\alpha (1+x)^\alpha - \beta \alpha x^\alpha -2 \beta \alpha x^{\alpha-1} + \beta x^{\alpha-1} + \alpha -1 \le 0 $. \\  \\
Since $ -2 \beta \alpha + \beta  < 0  \Rightarrow   -2 \beta \alpha + \beta  \ge -2 \beta \alpha x^{\alpha-1} + \beta x^{\alpha-1} $ for $ x \ge 1 $, we will be done with the proof if $q^\alpha (1+x)^\alpha - \beta \alpha  x^\alpha  -2 \beta \alpha + \beta + \alpha -1 \le 0$.  \\  \\
Let $ h(x) = q^\alpha (1+x)^\alpha - \beta \alpha x^\alpha  -2 \beta \alpha + \beta + \alpha -1 $. Therefore, we obtain that $ h'(x) = \alpha q^\alpha (1+x)^{\alpha-1} - \alpha^2 \beta x^{\alpha-1} = 0 \Rightarrow $
$ x = \frac{1}{2 \alpha^{\frac {1}{\alpha-1}} - 1} < 1 $. It shows that $h(x)$ is maximum at $ x =1 $ for $  x\ge 1$. The proof is completed by observing that
$h(1) = q^\alpha 2^\alpha - 3 \beta \alpha + \beta + \alpha - 1 = q^\alpha 2^{\alpha - 1} (2 - 3 \alpha) + \beta + \alpha - 1 = - q^\alpha 2^{\alpha - 1} (3 \alpha - 1) + \alpha - 1 < 0.$ \\ 

\noindent \textbf {Case 7: $ s_{a} > 0 , s_{o} = 0 , g_{0} < s* $} \\ \\
Observing that $\frac {d\phi} {dt_{0}}  = \beta (s*)^{\alpha-1} (-s_{a})$, we need to show that \\ \\
$ s_{a}^\alpha + g - \beta (s*)^{\alpha-1}s_{a} \le 0 $. \\ \\
Consider the case $\rho < s* $. In this case, $s_{a}$ is equal to $s*$. On substituting 
$s_{a} = s* $, the above inequality reduces to $\alpha \le \beta $. \\  \\
Next, consider the other case $\rho \ge s* $. From Lemma \ref{lemma7}, we obtain $s_{a} \le qg_{0}+qs* $. 
As $s_{a} \neq 0$, it is greater than or equal to $s*$. Therefore, $s* \le s_{a}\le qg_{0}+qs*$. We have already
proven for the case $s* = s_{a}$ above. Due to convexity, it only remains to prove the above inequality for 
$ s_{a}= qg_{0} +qs*$. Substituting $ s_{a}= qg_{0} +qs*$, we can rewrite the inequality above as
$ q^\alpha (g_{0} +s*)^\alpha + g - \beta (s*)^{\alpha-1} q (g_{0} +s*) \le 0 $.
Let $g_{0} = x s*$. (Note that $ x  \le 1 $). We can then write the inequality 
as $ q^\alpha (1+x)^\alpha -\beta q (1+x) + \alpha - 1 \le 0 $. As the left hand side of this inequality is a convex function and since $ 0 \le x \le 1 $, it suffices to prove the inequality for $x = 0$ and $x =1$. \\ \\
For x = 0, $ q^\alpha -  \beta q  + \alpha - 1 \le 0 \Leftrightarrow$ 
$ (q^\alpha - q \beta) + (-q \beta + \alpha - 1) \le 0 $ which can be easily seen. 
For x = 1, 
$ q^\alpha 2^\alpha - 2q \beta + \alpha - 1 \le 0 \Leftrightarrow$  
$ ( q^\alpha 2^\alpha - 2 \beta ) -2 (1- \frac{1}{\alpha}) \beta + \alpha -1 \le 0 \Leftrightarrow$ 
$ q^\alpha (2^\alpha - 2 . 2^{\alpha -1})  -  2 (1- \frac{1}{\alpha}) \beta + \alpha -1 \le 0 \Leftrightarrow$ 
$ -2 (1- \frac{1}{\alpha}) \beta + \alpha -1 \le 0 \Leftrightarrow \alpha \le 2 \beta $ which can again be seen easily. \\ 

\noindent \textbf {Case 8: $ s_{a} = 0 , s_{o} > 0 $} \\ \\
Note that the property of algorithm SqOA implies that 
$s_{a} =  0 \Rightarrow \rho <  s* \Rightarrow g_{0} < s* $.
Since $\frac {dE_{a}(t)} {dt} = 0$, $\frac {dE_{o}(t)} {dt} = s_{o}^\alpha + g$, and 
$\frac {d\phi} {dt} = \beta (s*)^{\alpha -1} s_{o}$, we need to show that the following inequality is true. \\ 
\\ 
$ \beta (s*)^{\alpha -1} s_{o} \le c s_{o}^\alpha + cg  $ \\ \\
Let $ s_{o} = x s* $. Then, the above inequality can be rewritten as $x \le x^\alpha + \alpha - 1 $ 
$ \Leftrightarrow x - x^\alpha \le \alpha - 1$.  
Differentiating $ x - x^\alpha $, we see that 
it is maximized at $ x = (\frac{1}{\alpha})^{\frac{1}{\alpha-1}} $. It suffices to show  that $(\frac{1}{\alpha})^{\frac{1}{\alpha-1}} (1 - \frac{1}{\alpha}) \le \alpha -1 $ 
$ \Leftrightarrow (\frac{1}{\alpha})^{\frac{1}{\alpha-1}} \le \alpha $ which is true since $\alpha > 1$.\\ 

\noindent \textbf {Case 9: $ s_{a} =s_{o} =0 $} \\ \\
As observed above, the property of algorithm SqOA implies that 
$s_{a} =  0 \Rightarrow \rho <  s* \Rightarrow g_{0} < s* $. Therefore, 
$ \frac{d\phi}{dt_{0}} = \beta (s*)^{\alpha - 1} (0 -  0) = 0$. Note that 
$\frac {dE_{a}(t)} {dt} = \frac {dE_{o}(t)} {dt} = 0$ since $s_{a} =s_{o}= 0$.
It completes the proof for this case. \\

\noindent \textbf {Case 10: $g_0 = s*$} \\ \\
Since $g_{0} = s*$, there exists an interval $[t_{0},t_{0}+dt_{0}]$ such that either 
$\hat{g_{0}} = \max \{s*,g_{0}\} = g_{0}$ or $\hat{g_{0}} = \max \{s*,g_{0}\} = s*$
at any point of time in $[t_{0},t_{0}+dt_{0}]$. 
In the former case, the proof corresponds to one of the cases $1, 2, 3, 4, 6$ (depending on the values 
of $s_{a}, s_{o}$ and $k$) because $\frac {d\phi} {dt}, \frac {dE_{a}(t)} {dt}, \frac {dE_{o}(t)} {dt}$
would exactly be the same in the corresponding case. In the later case, the proof corresponds to case $5$ or case $7$ 
(again depending on the values of $s_{a}$ and $s_{o}$) for the same reason.

\section{Conclusion}
In this paper, we have provided a new online algorithm SqOA which improves  the competitive ratio for energy from  $\{4, 2 + {\alpha}^{\alpha}\}$ (Han. et al. [8]) to max $\{4, 2 + (2-1/{\alpha})^\alpha 2^{\alpha-1}\}$. 
We believe that this problem requires further attention, as it is both academically interesting
and has practical applications. A particularly interesting research direction would be to analyse SqOA in a multiprocessor environment.  Another research direction can be to analyse SqOA is a speed-bounded model.

\end{document}